\documentclass[%
    aps,superscriptaddress,  % options for appearance
    pra, % note that prl style doesn't allow appendix and removes section numbering
    nofootinbib, %
    reprint,%twocolumn
    a4paper,% need to specify a4paper or letterpaper or watermark is messed up w/ revtex4
    longbibliography
    ]{revtex4-1} %

\usepackage{nicefrac}
\usepackage[english]{babel}
\usepackage{soul} % strikethrough text with \sta

\usepackage{array} % for m-columns in tabular
\usepackage{hyperref}
\usepackage{xcolor}
\definecolor{mylinkcolor}{rgb}{0,0,0.5} % set link color here as red,green,blue.
\hypersetup{unicode=true, %
  bookmarksnumbered=false,bookmarksopen=false,bookmarksopenlevel=1, %
  breaklinks=true,pdfborder={0 0 0},colorlinks=true}%
\hypersetup{%
  anchorcolor=mylinkcolor,citecolor=mylinkcolor, %
  filecolor=mylinkcolor,linkcolor=mylinkcolor, %
  menucolor=mylinkcolor,runcolor=mylinkcolor, %
  urlcolor=mylinkcolor}%
\usepackage{bbm} 
\usepackage{amsmath}
\usepackage{graphicx}
\usepackage{amsthm}
\usepackage{enumerate}
\usepackage{tikz}
\usepackage{amsfonts}

\AtBeginDocument{%
    \newwrite\bibnotes
    \def\bibnotesext{Notes.bib}
    \immediate\openout\bibnotes=\jobname\bibnotesext
    \immediate\write\bibnotes{@CONTROL{REVTEX41Control}}
    \immediate\write\bibnotes{@CONTROL{%
    apsrev41Control,author="08",editor="1",pages="1",title="0",year="0"}}
     \if@filesw
     \immediate\write\@auxout{\string\citation{apsrev41Control}}%
    \fi
  }%

\makeatletter
\def\pdfstartlink@attr{attr{/Border[0 0 0 [1 5] ]/H/I/C[0 1 1]}}%
\def\@@Doi#1{\textcolor{mylinkcolor}{#1}\@@endlink}
\makeatother
\def\ForTexCount\section#1{} % trick for subcounts in TexCount, without actually adding sections.

\newtheorem{thm}{Theorem} %[section]
\newtheorem{prop}{Proposition} %[section]
\newtheorem{lem}{Lemma} %[section]
\newtheorem{cor}{Corollary} %[section]
\theoremstyle{definition}
\newtheorem{ex}{Example} %[section]

\newcommand{\ket}[1]{| #1 \rangle}

\newcommand{\zero}{\uparrow}
\newcommand{\one}{\downarrow}

\newcommand{\cS}{\mathcal{C}}
\newcommand{\sta}{\mathcal{S}}

\newcommand{\cN}{\mathcal{N}}
\newcommand{\cQ}{\mathcal{Q}}
\newcommand{\cT}{\mathcal{T}}
\newcommand{\me}{\mathcal{M}}
\newcommand{\eff}{\mathcal{E}}

\newcommand{\ot}{\otimes}
\newcommand{\ext}{\mathrm{ext}}
\newcommand{\chsh}{\mathrm{CHSH}}
\newcommand{\ch}{\mathrm{CH}}
\newcommand{\eps}{\varepsilon}

\begin{document}

\title{Towards correlation self-testing of quantum theory in the adaptive
  Clauser-Horne-Shimony-Holt game}
\author{Mirjam Weilenmann}
\email{mirjam.weilenmann@oeaw.ac.at}
\affiliation{Institute for Quantum Optics and Quantum Information (IQOQI) Vienna,
Austrian Academy of Sciences, Boltzmanngasse 3, 1090 Vienna, Austria}
\author{Roger Colbeck}
\email{roger.colbeck@york.ac.uk}
\affiliation{Department of Mathematics, University of York, Heslington, York, YO10 5DD, United Kingdom}
\date{$15^{\text{th}}$ January 2024}

\begin{abstract} 
  Correlation self-testing of a theory addresses the question of whether we can identify the set of correlations realisable in a theory from its performance in a particular information processing task.  Applied to quantum theory it aims to identify an information processing task whose optimal performance is achieved only by theories realising the same correlations as quantum theory in any causal structure. In \href{\doibase 10.1103/PhysRevLett.125.060406}{[Phys.\ Rev.\ Lett.\ {\bf 125} 060406 (2020)]} we introduced a candidate task for this, the adaptive CHSH game. Here, we analyse the maximum probability of winning this game in different generalised probabilistic theories. We show that theories with a joint state space given by the minimal or the maximal tensor product are inferior to quantum theory, before considering other tensor products in theories whose elementary systems have various two-dimensional state spaces.  For these, we find no theories that outperform quantum theory in the adaptive CHSH game and prove that it is impossible to recover the quantum performance in various cases.  This is the first step towards a general solution that, if successful, will have wide-ranging consequences, in particular, enabling an experiment that could rule out all theories in which the set of realisable correlations does not coincide with the quantum set.
\end{abstract}
\maketitle

\section{Introduction}

Correlation self-testing means certifying that some information processing task (or set of tasks) can only be completed by a theory possessing a specific set of correlations~\cite{selftest_PRL}, for instance by a theory that has exactly the same correlations as quantum mechanics. Identifying such a task is of interest because it allows a new characterization of quantum theory as a theory that (within a set of probabilistic theories) is optimal for this task. This opens up the possibility of an experiment that can more directly exclude other theories because they lead to more restrictive correlations in the task at hand. Crucially, if an appropriate task is considered, theories that allow for more non-local bipartite correlations than quantum theory (such as box-world~\cite{Barrett2007}) can be inferior to quantum theory, in spite of the larger set of correlations realisable in the usual Bell scenario~\cite{selftest_PRL}.

Given a particular task we can consider the correlations that a given theory can realise in the scenario associated with that task. For instance, in a Bell experiment~\cite{Bell1964}, the CHSH value in general increases with the size of the bipartite state space~\cite{CHSH}.  In order to single out quantum theory, we need a task that, for a sufficiently large state space, becomes more difficult as the state space increases. An example of such a task is the adaptive CHSH game that we introduce in a companion paper~\cite{selftest_PRL}. It is a candidate for a task where theories that produce the correlations of quantum mechanics may perform optimally. 

This paper can be seen as a development of a line of research looking at ways to understand quantum mechanics from an information-theoretic perspective. The laws of quantum theory are usually introduced via a set of mathematical axioms, and there have been many attempts to understand the significance of these in other ways. One of the first insights was that requiring that signals cannot be transmitted faster than light is not sufficient for singling out quantum theory~\cite{Cirelson93,Popescu1994}. A variety of different theories do not allow superluminal signalling, for instance, generalised probabilistic theories (GPTs)~\cite{Hardy2001, Barrett2007}. Subsequently, identifying a principle that singles out quantum mechanics, has been a key question in quantum foundations. The proposals range from requiring communication complexity to be non-trivial~\cite{Brassard2006}, the impossibility of non-local computation~\cite{Linden2007} and the principle of \emph{Information Causality}~\cite{Pawlowski2009} (which imposes restrictions on random access coding) to principles such as Macroscopic Locality~\cite{Navascues2010a} and the multipartite principle of Local Orthogonality~\cite{Fritz2013b}. These principles are intended to represent properties of `reasonable' theories, but are insufficient to point uniquely to quantum correlations. In particular, the principles are obeyed by \emph{almost quantum correlations} (it remains possible that information causality is an exception, but numerical evidence suggests not)~\cite{Navascues2015}.

Correlation self-testing gives a new perspective on this problem. Instead of searching for a principle that excludes certain theories as implausible because of their implications, e.g.\ on information processing, we ask with respect to which aspects quantum mechanics is optimal. In particular, we seek to find information processing tasks for which quantum mechanics gives the optimal performance, while other theories are inferior.  Tasks with such properties are good candidates for self-tests of quantum theory.

After outlining self-testing in more detail and describing the framework of GPTs in Section~\ref{sec:prelim}, we proceed (see Section~\ref{sec:results}) to outline the adaptive CHSH game, whose winning probability is upper bounded by $\frac{3}{4}$ classically and $\frac{1}{2}+\frac{1}{2\sqrt{2}}$ quantum mechanically, proving that the latter bound is achievable. We then show that in any GPT where the joint state space is formed by taking the minimal or maximal tensor product this game cannot be won with a probability larger than with a classical strategy (Section~\ref{sec:tp}). To move beyond the minimal and maximal tensor product we consider systems in GPTs where the local state spaces are regular polygons in two dimensions with $n$ extremal states. For these we find that regular polygons with odd $n$ perform strictly worse than quantum theory (Section~\ref{sec:odd_n}) and we also prove the same for any GPT with local systems with $n=4$ (Section~\ref{sec:even_n}; for larger even $n$ the problem remains partially open).  In addition, we discuss the case of self-dual theories in Section~\ref{sec:sd}, for which we also find that they perform worse than quantum mechanics. Finally, in Section~\ref{sec:conclusion}, we conclude and discuss the remaining open problems and further research directions.

%%%%%%%%%%%%%%%%%%%%%%%%%%%%%%%%%%%%%%%%%%%%%%%%%%%%%%%%%%%%%%%%%%%%%%%%%%%%%%%%%%%%

\section{Preliminaries} \label{sec:prelim}

In this section we introduce the notion of correlation self-testing and its relation to other forms of self-testing as well as the tools needed in order to understand the remaining sections of this article, in particular the framework of GPTs.

\subsection{Self-testing of physical theories}
We consider an information processing task in which there are some number of (classical) inputs, labelled $X_i$ and (classical) outputs $A_i$ and we use ${\bf X}$ and ${\bf A}$ to denote the collection of all inputs and outputs respectively.  Given the setup for the task (which mathematically corresponds to constraints on the way the variables can be related) and a theory, particular conditional distributions $P_{{\bf X}|{\bf A}}$ can be realised.  An information processing task can be characterised by a function on the inputs and outputs, whose expectation value we want to maximize.  Alternatively, we can express this expectation value as a real-valued function $f$ on the conditional distributions.  We would like to find a setup, $\cS$, and a function, $f$, whose maximum over all realisable correlations in all theories\footnote{We do not say here what we mean by a theory because it makes sense to consider self-testing for various definitions. However, when we come on to the technical results of this paper, we will consider GPTs.}
can be achieved by quantum mechanics, i.e., we seek a case where (informally)
\begin{equation}
\max_{\{P_{{\bf X}|{\bf A}}\}}f(P_{{\bf X}|{\bf A}})=\max_{\{P_{{\bf X}|{\bf A}}\in\cQ(\cS)\}}f(P_{{\bf X}|{\bf A}})\,, \label{eq:self-test}
\end{equation}
where $\cQ(\cS)$ is the set of distributions realisable in quantum theory in the setup $\cS$, and the first maximum is over distributions realisable in $\cS$ in any theory.
Furthermore, we would like it to be the case that quantum mechanics outperforms as many alternative theories as possible, ideally being the unique theory that can achieve this maximum within the set of physical theories under consideration.  If so, this task provides a way to \emph{self-test} quantum theory within this set of theories.  Given such a task, an experiment could be performed showing that quantum mechanics is the only theory compatible with the observed results, thus providing a much more direct test of quantum theory than ever before. 

The set of theories we consider in this work are the GPTs (see Section~\ref{sec:gpt} for more details). Achieving complete self-testing of quantum theory in this set is not possible, because several theories can give rise to the same sets of correlations.\footnote{When speaking about the correlations of a theory we implicitly refer to the closure of the set of correlations.} For instance, quantum theory over the real numbers gives the same correlations as complex quantum mechanics in many scenarios~\cite{MMG}. The same holds if we aim to self-test specific aspects of quantum theory rather than the theory as a whole, for instance its state space. However, the above task is still of interest, since being able to self-test quantum theory in a certain subset of all GPTs (that do not lead to the same correlations), could be useful as a way  to rule out some of these theories experimentally. This may for instance be of interest if there is a particular subset of GPTs with natural properties that are competing with quantum theory. 

Due to the aforementioned restrictions on self-testing from all conceivable theories, we consider the task of \emph{correlation self-testing}. Applied to quantum theory, this related but weaker question is whether we can find a task where any theory capable of achieving the optimum has the same set of realisable correlations as quantum theory.  In principle, there could be several different theories (including quantum mechanics) that can achieve the same optimum in that task. If all such theories necessarily allow the same set of realisable correlations in any task, then we are able to correlation self-test quantum theory using the task at hand.  In other words, with respect to the set of theories under consideration, correlation self-testing of quantum theory (with a single task $f$) means that if~\eqref{eq:self-test} holds and a theory gives rise to correlations $\cT(\cS)$ such that
\begin{equation}\label{eq:Qopt}
\max_{\{P_{{\bf X}|{\bf A}}\in \cT(\cS) \}}f(P_{{\bf X}|{\bf A}})=\max_{\{P_{{\bf X}|{\bf A}}\in \cQ(\cS) \}}f(P_{{\bf X}|{\bf A}}),
\end{equation}
then we have $\cT(\cS')=\cQ(\cS')$ in all causal structures $\cS'$.\footnote{In the context of this work, these are diagrams indicating the allowed flow of information. For a more technical account see, e.g.,~\cite{review}.} In other words, correlation self-testing of quantum theory with a single task $f$ means that, if for a theory $\cT$ there is a setup $\cS$ where $\cT(\cS)\neq\cQ(\cS)$, then~\eqref{eq:Qopt} does not hold and the left-hand side is smaller than the right.

Correlation self-testing quantum theory with a set of tasks $\{f_i\}_i$ is possible if whenever another theory leads to a different set of correlations in some scenario, then for one of the tasks in $\{f_i\}_i$, quantum theory outperforms it. An important step towards correlation self-testing is thus to identify tasks where quantum theory outperforms other known theories. Note that finding such tasks is not only relevant for correlation self-testing but also for GPT self-testing in certain sets of theories as introduced above.

In this work, we make a step towards correlation self-testing by asking for which tasks, $f$, it holds that 
\begin{equation} 
  \max_{\{P_{{\bf X}|{\bf A}}\in \cT(\cS) \}}f(P_{{\bf X}|{\bf A}}) \leq \max_{\{P_{{\bf X}|{\bf A}}\in \cQ (\cS)\}}f(P_{{\bf X}|{\bf A}}) \,, \label{eq:self-test2}
\end{equation}
where $\cT(\cS)$ is the set of correlations a theory can achieve in the corresponding causal structure, $\cS$, and where the inequality is strict for as many theories as possible. Specifically, we analyse the adaptive CHSH game (see Section~\ref{sec:results}) for a variety of theories, establishing that~\eqref{eq:self-test2} holds for the theories considered and proving in the majority of analysed cases that the inequality is strict.

\subsection{Generalised probabilistic theories}  \label{sec:GPT} \label{sec:gpt}
From an operational viewpoint, any physical experiment consists of the preparation of a system, a transformation of it and a subsequent measurement that yields some outcomes. Mathematically, we model a system's state space $\sta$ using a compact convex subset of a real vector space $V$.\footnote{The elements of the vectors are related to the probabilities of measurement outcomes, and convexity is justified by the notion that it is possible to prepare a system in state $S_1$ with probability $p$ or $S_2$ with probability $(1-p)$ and then forget the preparation, resulting in $pS_1+(1-p)S_2\in\sta$.}

For each type of system there is an associated effect space, $\eff$, a compact convex subset of the dual vector space to $V$ (denoted $V^*$) made up of linear maps ${\sta\to[0,1]}$ called effects. For a system with state space $\sta \subset V$, the maximal possible effect space $\eff_\mathrm{max}$ is,
 \begin{align}
 \eff_\mathrm{max}(\sta)= \left\{ e \in V^* \mid 0 \leq e(S) \leq 1 \ \forall S \in \sta \right\}.
 \end{align}
 Any effect space $\eff$ contains a \emph{unit effect} $u \in\eff$, which has the property that $u(S)=1$ for all $S\in\sta$. If not stated otherwise, we consider the maximal effect space to any state space, since bounds on the performance in the adaptive CHSH game with the maximal effect space then remain valid under restrictions on the effect space. This means that although we use the \emph{no-restriction hypothesis}~\cite{NoR1, NoR2}, our results are not restricted to theories that satisfy it.  
 
 A measurement $M \in \me$ is a collection of effects that sum to the unit effect, i.e., we can write $M=\{e^x\in\eff:\sum_xe^x=u\}$.  The interpretation of $e^x(S)$ is the probability of outcome $x$ when $M$ is performed on a system in state $S$.

Transformations may be applied to a system between preparation and measurement. These are linear maps from the state space to itself that for each type of system include an identity transformation, $I$. Since a transformation followed by a measurement is equivalent to a different measurement, we will not need to explicitly consider the set of possible transformations.

In principle, states in a theory could require an infinite number of parameters to specify. For instance, this would be the case if the theory had an infinite number of different measurements whose outcomes are unrelated to one another.  However, in many theories it is possible to infer the outcome probabilities of some measurements from those of others. A collection of measurements that (for any state of a particular system) allows us to infer the outcome probabilities for all other measurements is called a \emph{set of fiducial measurements} for that system. Here, we consider only theories where the state of each system can be completely characterised in terms of a finite number of fiducial measurements with a finite number of outcomes each, so a finite vector is sufficient to represent the state.\footnote{Some representations have redundancies, e.g., if their components contain all the outcome probabilities of a particular measurement then there is a redundancy due to normalisation; likewise there are redundancies if states are written in terms of more measurements than a minimal set of fiducial measurements.} The extremal states of $\sta$ are called \emph{pure}, in analogy to the quantum case. Note that certain such theories, for instance  classical physics, have systems with state spaces that have a finite number of pure states, while for others, such as quantum theory, the number of pure states is infinite.

A system whose state space can be characterised by two two-outcome fiducial measurements is called a \emph{gbit}. The state space of a gbit can be written as a convex shape in $\mathbb{R}^2$, i.e., states can be represented by vectors $(P(0|0),P(0|1))$ (the (redundant) probabilities of obtaining outcome $1$ for each fiducial measurement are not needed). Depending on the number of pure states there are different polygonal state spaces (including all cases with finitely many pure states) and state spaces with (partly) curved boundary (for infinitely many pure states).  We consider regular polygons, which have the property that for any two pure states there is a reversible transformation that maps one to the other~\cite{JanottaThesis}.\footnote{Truncated versions of these polygon state spaces may also have this property~\cite{JanottaThesis}.} In a classical system all pure states are perfectly distinguishable and the state space is a simplex. In 2D, this is a trit with $n=3$. The following example illustrates the state space for $n=4$.

\begin{ex} \label{ex:gbit} Consider the state space $\sta_{\Box}$ spanned by the four extremal states $(0,0)$, $(0,1)$, $(1,0)$ and $(1,1)$, i.e., the state of a gbit in this theory lies within this space where each co-ordinate represents the probability of outcome $0$ for one of two fiducial measurements. This corresponds to a system where there can at the same time be complete certainty about the outcomes of both fiducial measurements (this situation differs qualitatively from the quantum case).
\end{ex}

\subsubsection{Bipartite systems} \label{sec:bipartite}
Two systems $A$ and $B$ can be thought of as parts of a single joint
system $AB$.  A minimal requirement on the joint state space is that if $S_A\in\sta_A$ and
$S_B\in\sta_B$ then $S_A\ot S_B\in\sta_{AB}$. States that can be written as $S_A\ot S_B$ are called \emph{product} and convex combinations thereof are \emph{separable}, while all other states are called \emph{entangled}. Analogously, measurements $M=\{e^x\}_{x}\in\me_A$ and
$N=\{f^y\}_{y}\in\me_B$ can be composed into
$L\in\me_{AB}$, made up of \emph{product effects}, i.e.,  $L=\{e^x\ot f^y\}_{x,y}$. Product effects act on product states as ${(e^x \ot f^y)(S_A\ot S_B)}=e^x(S_A)f^y(S_B)$. This implies that the theories we consider are non-signalling, since if $\{e_a^x\}_x\in\me_A$ and $\{f_b^y\}_y\in\me_B$ are measurements for $a=1,\ldots,n_a$ and $b=1,\ldots,n_b$, then
\begin{align*}
p(y|a,b)
\!= \! \sum_x(e_a^x\ot f_b^y)(S_{AB}) \! = \! (u_A\ot f_b^y)(S_{AB}) \! = \! p(y|b)\,,
\end{align*}
which is independent of $a$, and analogously for the other no-signalling conditions. Note that this premise is non-trivial and may exclude some theories. For instance, formulations of quantum theory over the quaternions have difficulties with defining tensor products~\cite{AaronsonBlog, Baez}. It is, however, natural to impose no-signalling in order to allow such theories to be compatible with special relativity.

Performing a measurement $M=\{e^x\}_{x}\in\me_A$ on a subsystem $A$ of a multi-partite system $AB$, means applying the maps $e^x \otimes I_B$ to $S_{AB}$. The post-measurement state on $B$ after observing outcome $x$ is
\begin{equation} \label{eq:condstate}
S_{B|x}= \frac{(e^x\ot I_B)(S_{AB})}{e^x(S_A)}\,,
\end{equation}
where $S_A=(I_A\ot u_B)S_{AB}$.

Given a bipartite state $S_{AB}$ and a set of measurements $\cN_A\subset\me_A$ on $A$ and $\cN_B\subset\me_B$ on $B$ we can construct the joint distribution of outcomes.  These can be written as a vector, which we denote $V(S_{AB},\cN_A,\cN_B)$, whose entries can be conveniently displayed in a matrix.  For instance, if we consider two binary measurements on a bipartite system then we can represent $V(S_{AB},\cN_A,\cN_B)$ as\footnote{Although written as a matrix, we consider this a way to display the elements of the vector.}
\begin{align*}
	\left(
	\begin{array}{cc|cc}
	P(00|00)&P(01|00)&P(00|01)&P(01|01)\\
	P(10|00)&P(11|00)&P(10|01)&P(11|01) \\
	\hline
	P(00|10)&P(01|10)&P(00|11)&P(01|11)\\
	P(10|10)&P(11|10)&P(10|11)&P(11|11) 
	\end{array}
	\right).
\end{align*}

In this work, we impose \emph{local tomography} for the GPTs under consideration, meaning that the joint state of a multi-partite system can be fully characterised by the statistics obtained from local measurements on its component subsystems. This ensures that systems can be completely described in terms of the outcome probabilities of all combinations of product fiducial measurements on the subsystems. (For theories like real quantum mechanics, where this is not the case, additional global parameters are necessary for characterising a state.)  In locally tomographic theories, the two extremal choices of a joint state space of two single-system state spaces are the \textit{minimal} and the \textit{maximal tensor product} respectively. The minimal tensor product, $\otimes_{\rm min}$, of two state spaces is the convex hull of all possible product states, i.e., it comprises all separable states.  For theories where the local effect space is $\eff_{\rm max}(\sta)$, the maximal tensor product of two state spaces, $\otimes_{\rm max}$, is defined as the set of joint states, $S$, for which for all product effects $e$ we have $0\leq e(S)\leq1$, and such that the tensor product of the single system unit effects acting on $S$ gives $1$. Generically these feature entangled states (but classical cases do not). When the local effect space is not $\eff_{\rm max}(\sta)$ a more general definition of $\otimes_{\rm max}$ is needed (see Section~\ref{sec:sd}).  If the state space of a single system is that given in Example~\ref{ex:gbit}, taking the minimal and maximal tensor product leads to general local theory and box-world respectively, both of which were considered in~\cite{Barrett2007}.

We are not aware of work that provides constructions of (explicit) joint state spaces between the minimal and the maximal tensor product, with the exception of quantum mechanics. However, for such a joint state space to be considered reasonable, we require it to be convex, since we could have a preparation procedure where we randomly prepare one of two states and then forget which one it was.  We also require that for any state and set of local measurements, if the local outcome probabilities are permuted, then there is a state that achieves these permuted correlations under the same measurements. This imposes a symmetry on the joint state space.

\section{Correlation self-testing in the  adaptive CHSH game} \label{sec:results}

We consider the \emph{adaptive CHSH game} introduced in~\cite{selftest_PRL}. A referee asks Bob to choose one of four permutations of the CHSH game~\cite{CHSH}. He then asks Alice and Charlie questions, denoted $R_A$ and $R_C$ respectively, for which they have to give answers, labelled $A$ and $C$ respectively, where $R_A$, $R_C$, $A$ and $C$ can take values in $\left\{0,1 \right\}$. The questions are chosen uniformly at random. The three players win the game if Alice's and Charlie's answers win the CHSH-type game Bob chose. 

Bob can express his choice in terms of two bits, which correspond to the winning conditions in the following table. 
\vspace{-0.5cm}
\begin{center} \begin{tabular}{| c | c |} 
		\hline
		Bob's choice $B$ & Winning condition of CHSH-type game\\
		\hline
		$(0,~0)$ & $r_A \cdot r_C \oplus r_C = a \oplus c$ \\
		\hline
		$(0,~1)$ & $r_A \cdot r_C \oplus r_A \oplus r_C = a \oplus c$ \\
		\hline
		$(1,~0)$ & $r_A \cdot r_C \oplus r_A \oplus r_C \oplus 1 = a \oplus c$ \\
		\hline
		$(1,~1)$ & $r_A \cdot r_C  \oplus r_C \oplus 1= a \oplus c$ \\
		\hline
	\end{tabular}
\end{center}

The three players are aware of the rules of the game and may each choose a strategy. Before the game starts Alice may share a resource $S_{AB}$ with Bob, and Charlie may share $S_{B'C}$ with Bob, where the nature of $S_{AB}$ and $S_{B'C}$ depends on the theory we are working in. However, no other resource is shared between the parties (in particular, no tripartite system). Furthermore, after the game commences, no communication is allowed, so, in particular, Alice and Charlie are not aware of Bob's choice when they give their outcomes. The causal structure associated with the game is shown in Figure~\ref{fig:game}.

Alice, Bob and Charlie may have a coordinated strategy, which one might think of as an additional tripartite (classical) cause. However, because their strategy is fixed, we can think of our result applying for all fixed strategies of Alice, Bob and Charlie without the need for this additional shared resource.

The overall winning probability of the game for a strategy that leads to a distribution ${P_{ABCR_AR_C}(a,b,c,r_A,r_C)}=\frac{1}{4} P_{ABC|R_AR_C}(a,b,c | r_A, r_C)$ is 
\begin{align}
p_\mathrm{win}(P)  = \! \! \! \! \! \! \! \! \! \! \sum_{a,b,c,r_A,r_C} \! \! \! \! \! \! P_{ABCR_AR_C}(a,b,c,r_A,r_C) Q(a,b,c,r_A,r_C),
\end{align} 
where $Q=1$ if the winning condition specified in the above table is met and $Q=0$ otherwise.
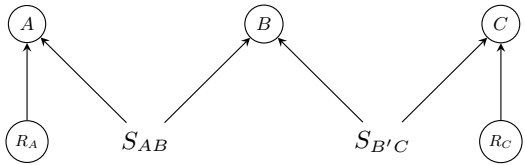
\begin{figure}
	\centering
	\resizebox{0.8\columnwidth}{!}{%
		\begin{tikzpicture} [scale=1.5]
		\node[draw=black,circle,scale=0.75] (1) at (0,1) {$A$};
		\node[draw=black,circle,scale=0.65] (R1) at (0,0) {$R_A$};
		\node (2) at (1,0) {$S_{AB}$};
		\node[draw=black,circle,scale=0.75] (4) at (2,1) {$B$};
		\node (5) at (3,0) {$S_{B'C}$};
		\node[draw=black,circle,scale=0.75] (6) at (4,1) {$C$};
		\node[draw=black,circle,scale=0.65] (R2) at (4,0) {$R_C$};
		\draw [->,>=stealth] (R1)--(1);
		\draw [->,>=stealth] (2)--(1);
		\draw [->,>=stealth] (2)--(4);
		\draw [->,>=stealth] (5)--(4);
		\draw [->,>=stealth] (5)--(6);
		\draw [->,>=stealth] (R2)--(6);
		\end{tikzpicture}
	}
	\caption{Causal structure of the adaptive CHSH game. The answers $A$, $B$ and $C$ as well as the referee's questions $R_A$ and $R_C$ determine whether the game is won. The resources $S_{AB}$ and $S_{B'C}$ are modelled as unobserved resources in the causal structure language, and their nature depends on the theory.}
	\label{fig:game}
\end{figure}

\subsection{Optimal quantum strategy} \label{sec:quantum}
If the players have access to quantum systems, i.e., if Alice can share an entangled quantum system $\rho_{{AB}}$ with Bob, and Charlie can share $\rho_{{B'C}}$ with Bob, then the winning probability for the optimal quantum strategy $P_Q^{\mathrm{opt}}$ is $p_{\mathrm{win}}(P_Q^{\mathrm{opt}})=\frac{1}{2}+\frac{1}{2\sqrt{2}}$. That this is an upper bound follows from Tsirelson's bound~\cite{Cirelson}, and in the following we show that this bound is achievable with quantum states and measurements. 

We express the strategy using the Bell states
\begin{eqnarray*}
\ket{\Psi_{00}}&=&\frac{1}{\sqrt{2}}(\ket{\!\zero\zero}+\ket{\!\one\one})\\
\ket{\Psi_{01}}&=&\frac{1}{\sqrt{2}}(\ket{\!\zero\zero}-\ket{\!\one\one})\\ \ket{\Psi_{10}}&=&\frac{1}{\sqrt{2}}(\ket{\!\zero\one}+\ket{\!\one\zero})\\ \ket{\Psi_{11}}&=&\frac{1}{\sqrt{2}}(\ket{\!\zero\one}-\ket{\!\one\zero})\,.
\end{eqnarray*}
The quantum states shared are $\ket{\Psi_{00}}_{AB}$ and $\ket{\Psi_{00}}_{B'C}$. Bob makes his choice by measuring the system $BB'$ in the Bell basis, i.e., $\{\ket{\Psi_{00}}_{BB'},\ket{\Psi_{01}}_{BB'},\ket{\Psi_{10}}_{BB'},\ket{\Psi_{11}}_{BB'}\}$. He gives outcome $B\in\{(0,0),\,(0,1),\,(1,0),\,(1,1)\}$ using the indices of the state corresponding to the outcome obtained.

Defining $\ket{\theta}=\cos(\frac{\theta}{2})\ket{\!\zero}+\sin(\frac{\theta}{2})\ket{\!\one}$, the measurements generating $A$ and $C$ are as follows:
\begin{itemize}
	\item if $R_A=0$, measure in the $\{\ket{0},\ket{\pi}\}$ basis,
	\item if $R_A=1$, measure in the $\{\ket{\pi/2},\ket{3\pi/2}\}$ basis,
	\item if $R_C=0$, measure in the $\{\ket{\pi/4},\ket{5\pi/4}\}$ basis,
	\item if $R_C=1$, measure in the $\{\ket{3\pi/4},\ket{7\pi/4}\}$ basis.
\end{itemize}
For each measurement, if the first element in the basis is obtained, the outcome is set to $0$ and otherwise it is set to $1$.

The joint probability distribution obtained when $B=(0,0)$ is\bigskip\\
\begin{tabular}{cc|cc|cc|}
	$P_{AC|R_AR_C,B=(0,0)}$&&$R_C=0$&&$R_C=1$&\\
	&&$C=0$&$C=1$&$C=0$&$C=1$\\
	\hline
	\vphantom{$\frac{\sum}{f}$}$R_A=0$&$A=0$&$\frac{1+\eps}{4}$&$\frac{1-\eps}{4}$&$\frac{1-\eps}{4}$&$\frac{1+\eps}{4}$\\
	\vphantom{$\frac{\sum}{f}$}&$A=1$&$\frac{1-\eps}{4}$&$\frac{1+\eps}{4}$&$\frac{1+\eps}{4}$&$\frac{1-\eps}{4}$\\
	\hline
	\vphantom{$\frac{\sum}{f}$}$R_A=1$&$A=0$&$\frac{1+\eps}{4}$&$\frac{1-\eps}{4}$&$\frac{1+\eps}{4}$&$\frac{1-\eps}{4}$\\
	\vphantom{$\frac{\sum}{f}$}&$A=1$&$\frac{1-\eps}{4}$&$\frac{1+\eps}{4}$&$\frac{1-\eps}{4}$&$\frac{1+\eps}{4}$\\
	\hline
\end{tabular}\,,\bigskip\\
where $\eps=1/\sqrt{2}$.  This corresponds to a winning probability of $\frac{1}{2}\left(1+\frac{1}{\sqrt{2}}\right)$. An analogous distribution is obtained for the other values of $B$ and the winning probability is the same for each possible $B$. Note that although Alice and Charlie never know the value of $B$, it is correlated with the quantum state they hold.

\subsection{Theories whose joint state space is given by either the minimal or maximal tensor product}\label{sec:tp}
Taking the minimal or maximal tensor product to form the joint state space is equivalent to taking either the largest state space, or largest effect space compatible with the single system state space of a locally tomographic theory. In this section we show that regardless of the single-system state space of a GPT, whenever the joint state space is given either by the minimal or by the maximal tensor product, then the theory is unable to outperform classical physics in the adaptive CHSH game.

\begin{prop} \label{thm:minmax} 
In any GPT where the joint state space is either defined as the minimal or as the maximal tensor product of the individual state spaces, the optimal winning probability in the adaptive CHSH game is bounded by $3/4$.
\end{prop}
The proof of this proposition is given in Appendix~\ref{app:minmax}.  The idea is that in such theories, any strategy can be classically simulated, and the winning probability for any classical strategy is at most $3/4$. Note that previous work~\cite{Short2006,Skrzypczyk2009b,Skrzypczyk2009} implies this for boxworld.

This bound can be achieved in any theory in which there are (at least) two states and a measurement $M$ that perfectly distinguishes them, i.e., gives the outcome $0$ for one state and $1$ for the other with unit probability.  To see this, take the sources $S_{AB}$ and $S_{B'C}$ to share with uniform probability either the first or the second of these perfectly distinguishable states and let each party perform the measurement $M$ on each of the obtained systems. Let Alice and Charlie output their measurement outcome 
 respectively. Then Bob outputs $(0,0)$ or $(1,0)$ each with probability $\frac{1}{2}$ if he obtains the same outcome from both his measurements and he outputs $(0,1)$ or $(1,1)$ each with probability $\frac{1}{2}$ otherwise.

\subsection{Theories with two-dimensional single-system state spaces} \label{sec:polygons}
In this section we make progress on analysing their performance in the adaptive CHSH game for  theories beyond those whose joint state space is given by the minimal or maximal tensor product. To do so we consider GPTs whose local state spaces are regular polygons.  Because we lack the tools to construct different bipartite state spaces for these theories explicitly, we instead generate bounds that apply to any theories with such local state spaces.  It turns out that the cases of regular polygons with an odd number of vertices differ qualitatively from the cases for which this number is even and that self-dualization of the local state spaces leads to yet another behaviour. We consider each of these separately below.

\subsubsection{Regular polygon state spaces with an odd number of extremal vertices}\label{sec:odd_n}
In the cases with an odd number of vertices we argue in the following  that it is impossible to win the adaptive CHSH game with a probability larger than quantum theory because it is impossible to win any CHSH game with probability larger than quantum theory with such a local  state space. For the usual CHSH game, the larger the state space, the larger the winning probability. Hence, to upper bound the winning probability of the CHSH game we can, in this case, consider the joint state space formed with the maximal tensor product.

For $n=5,7,\ldots,29$ we used a linear program to show that no state in the maximal tensor product of two such systems can lead to correlations that exceed Tsirelson's bound for any CHSH game, and hence also in the adaptive CHSH game. In fact, the correlations that we obtained are strictly below Tsirelson's bound in all cases.

To see this, consider the dual cone to the corresponding local state space, $\mathcal{E}_{\ext}$, which has $n$ rays and $2n$ extremal effects (for each effect, $e$, on a ray, $u-e$ is a (non-ray) extremal effect)~\cite{Janotta2011}. Given a bipartite state $S_{AB}$, a set $\cN_A$ of two binary measurements on $A$, and another such set $\cN_B$ on $B$, we can compute the joint distribution of the outcomes and arrange it as a vector, $P=V(S_{AB},\cN_A,\cN_B)$. Since the CHSH value is a linear function of $P$, we can compute this by taking the inner product of this vector with another vector $C_{\chsh}$. Due to linearity, the maximal CHSH value is furthermore obtained by performing extremal measurements on extremal states.

Now, for any choice of extremal measurements $\cN_A$ and $\cN_B$, we can obtain the maximal CHSH value for this choice by the following optimisation:
\begin{align*}
\max_{S_{AB}} \quad &C_{\chsh}.V(S_{AB},\cN_A,\cN_B) \\
\operatorname{subj. to} \quad &(e_A \otimes e_B) (S_{AB}) \geq 0 \ \ \forall e_A, e_B \in \mathcal{E}_\ext\\
& (u_A\ot u_B)(S_{AB}) = 1\,,
\end{align*}
where the last two conditions correspond to restricting $S_{AB}$ to be a valid state.

Comparing the results for all choices of four measurements, we obtain the optimal CHSH values summarised in the following table.

\begin{widetext}
{\small
	\begin{center} \begin{tabular}{| c | c | c |} 
			\hline
			$n$ & optimal CHSH value & formulaic description of value\\
			\hline
			5 & 0.8354
			& $\frac{\vphantom{\sum^A}1}{\vphantom{\frac{1}{f}}2} + \frac{1}{4(1+\sec(\pi/n))^2}(-1 +
			\sec(\pi/n)(2\cos(\frac{3 + n}{4n}\pi) + 6\sin( \frac{1 + n}{4n}\pi) - \sec(\pi/n) + 2))$ $ \  \vphantom{*}$ \\
			\hline
			7 & 0.8462  & $\frac{\vphantom{\sum^A}1}{\vphantom{\frac{1}{f}}2} + \frac{1}{4(1 + \sec(\pi/n))^2} (1 + 
			\sec(\pi/n)(2 \sin( \frac{3 + n}{4n}\pi) + 
			6\cos(\frac{1 + n}{4n}\pi) + \sec(\pi/n) - 2))$ *\\
			\hline
			9 & 0.8497	& $\frac{\vphantom{\sum^A}1}{\vphantom{\frac{1}{f}}2} + \frac{1}{4(1 + \sec(\pi/n))^2} (1 + 
			\sec(\pi/n)(2 \cos( \frac{3 + n}{4n}\pi) + 
			6\sin(\frac{1 + n}{4n}\pi) + \sec(\pi/n) - 2))$ *\\
			\hline
			11 & 0.8505 & $\frac{\vphantom{\sum^A}1}{\vphantom{\frac{1}{f}}2} + \frac{1}{4(1 + \sec(\pi/n))^2} (-1 +
			\sec(\pi/n) (
			2\sin(\frac{3 + n}{4n}\pi) + 6 \cos( \frac{1 + n}{4n}\pi)   - \sec(\pi/n) + 2))$ $ \  \vphantom{*}$ \\
			\hline
%			13 & 0.851569 & $\frac{\vphantom{\sum^A}1}{\vphantom{\frac{1}{f}}2} + [\frac{-2}{(1 + \sec(\pi/n))^2} (1 - 
%			\sec(\pi/n)(6 \sin( \frac{1 + n}{4n}\pi) + 
%			2\cos(\frac{3 + n}{4n}\pi) - \sec(\pi/n) + 2))]/8$ $ \  \vphantom{*}$ \\
%			\hline
%			15 & 0.852039 & $\frac{\vphantom{\sum^A}1}{\vphantom{\frac{1}{f}}2} + [\frac{2}{(1 + \sec(\pi/n))^2} (1 + 
%			\sec(\pi/n)(2 \sin( \frac{3 + n}{4n}\pi) + 
%			6\cos(\frac{1 + n}{4n}\pi) + \sec(\pi/n) - 2))]/8$ *\\
%			\hline
	\end{tabular} \end{center} 
}
\end{widetext}

 Although our table involves fixed $n$, we give a formulaic description of how to generate the value that depends on $n$.  The reason for doing so is that our numerical results suggest that these formulae repeat such that the formula depends on the value of $n \! \mod 8$ (a similar repetition was suggested in~\cite{Janotta2011}). In general, formulaic descriptions of this kind can be extracted from the solution to the linear program above by retaining the information with which combinations of effects we have obtained the optimal CHSH value and then looking at the solution of the linear program which yields the corresponding state $S_{AB}$.  Writing $S_{AB}$ as an analytic expression leads to formulae like above (but where $n$ is the particular number of extremal states). We have checked numerically that these formulae are accurate up to $n= 29$.

 Our results relate to those of~\cite{Janotta2011}, where for one particular extremal state in the maximal tensor product state space the maximal CHSH value was computed. For the values $n=7,9,15,17,23,25$, our optimisation recovers the values from~\cite[Table~A2]{Janotta2011} (indicated in the table above with a $*$), while for $n=5,11,13,19,21,27,29$ we obtain larger values than those of~\cite[Table~A2]{Janotta2011}.  Note, however, that all of our optimal values can be recovered with measurements on the particular state that was analysed in~\cite{Janotta2011}, and we hence suspect that the differences in the formulaic description originates from typos in~\cite[Table~A2]{Janotta2011}.

Our results imply that if the systems held by Alice and Charlie have a single-system state space that is a regular polygon with $n=5,7,\ldots,29$ vertices, the maximum winning probability in the adaptive CHSH game is strictly smaller than in quantum theory. We do not exclude, however, that the use of several such systems could beat the quantum bound.  For instance, it could be that the scenario of the adaptive CHSH game can lead to several copies of a particular state shared between Alice and Charlie, and that from such copies it is possible to beat Tsirelson's bound. This would correspond to the ability to distil non-locality from such states.  Unfortunately, non-locality distillation is not well-understood and whether this is possible is an open question. For further details regarding this issue, we refer to Section~\ref{sec:distill}.

We further remark that if we restrict the adaptive CHSH game to be played with one pair of gbits shared by each source, then our results imply that all theories with the local state spaces we have analysed here perform strictly worse than quantum mechanics in this game. We can think of the adaptive CHSH game with a restriction on the dimension of the shared systems as \emph{dimension-dependent self-testing}.

\subsubsection{Regular polygon state spaces with an even number of extremal vertices} \label{sec:even_n}
The ability to construct and analyse explicit theories in which the joint state space is given by a tensor product other than the minimal or maximal tensor products would be especially useful for considering theories where the single-system state space is a regular polygon with an even number of vertices, since in this case it is known that using the maximal tensor product leads to correlations that can be distilled to a PR-box. Understanding to what extent this is also the case for other tensor products is important for upper bounding the performance of these theories in the adaptive CHSH game. Furthermore, an explicit construction of other tensor products would permit us to fully determine their performance.

In spite of the lack of explicit constructions, we can nevertheless derive bounds on the performance of such theories in the adaptive CHSH game. To do so, we analyse the performance of theories with joint state spaces of two gbits that have a CHSH value restricted to some specific value.  There are eight CHSH-type inequalities (related to one another by symmetries)~\cite{Cirelson93}, and hence it makes sense to restrict all of them simultaneously.  We express the quantities involved in these inequalities by $J_{\chsh}^i(P)$, where $i$ runs from 1 to 4 (these quantities are in direct correspondence with the winning conditions in the adaptive CHSH game), and $P$ is a probability distribution with two inputs and two outcomes. We get eight inequalities via lower and upper bounds, for instance we can take $J_{\chsh}^1(P)$ to be
\begin{align*}
  \frac{1}{4}\big(&P(11|01)+P(00|01)+P(11|00)+P(00|00)\\
  &+P(11|11)+P(00|11)+P(10|10)+P(01|10) \big)\,,
\end{align*}
so that two of the CHSH inequalities can be written $1/4\leq J_{\chsh}^1(P)\leq3/4$.

Restricting the CHSH value of a theory to satisfy the classical bound hence corresponds to saying that for all bipartite states $S_{AB}$ of two gbits in the theory, and pairs of measurements $\cN_A=(M_1,M_2)$ on Alice and $\cN_B=(N_1,N_2)$ on Bob (here $M_i\in\me_A$ and $N_i\in\me_B$), the distribution $P=V(S_{AB},\cN_A,\cN_B)$ obeys 
$1/4\leq J_{\chsh}^i(P)\leq3/4$ for $i=1,2,3,4$.

The following theorem gives a bound for the adaptive CHSH game for $n=4$, i.e., for systems where the single-system state space is $\sta_{\Box}$.
\begin{thm} \label{thm:gbits}
Let $\epsilon\in[0,1/8]$ and consider a GPT with single-system state space $\sta_{\Box}$ and any  joint state space~\footnote{Recall form \ref{sec:GPT} that a joint state space has to be convex and symmetric under permutations.}  where the CHSH values are bounded by ${2}{\epsilon} \leq J_{\chsh}^i(P) \leq 1-{2}{\epsilon}$ for $i=1,2,3,4$. Let the sources $S_{AB}$ and $S_{B'C}$ each correspond to one pair of gbits. Then the winning probability in the adaptive CHSH game is bounded by
\begin{equation} \label{eq:gbitbound}
p_{\mathrm{win}} \leq \frac{3}{4} + \frac{\epsilon  (1-8 \epsilon)}{96 \epsilon ^2-12 \epsilon +1} .
\end{equation}
\end{thm} 

\begin{proof}
  The restrictions on the CHSH values are equivalent to the following restrictions on the correlations in the bipartite state space:\footnote{Due to the redundancy in the description of states there are several equivalent formulations of these $8$ constraints. Note that for $\epsilon=\frac{1}{8}$ we recover the state space of the minimal tensor product, and for $\epsilon=0$ that of the maximal tensor product.}  {\small\begin{align}
4\epsilon \! - \! \frac{1}{2} & \leq P(10|01)+P(11|10)+P(01|11) \! - \! P(11|00)  \leq \frac{3}{2}  \! - \! 4\epsilon \nonumber \\
4\epsilon  \! - \!  \frac{1}{2} & \leq P(11|01)+P(01|10)+P(10|11)  \! - \! P(11|00)  \leq \frac{3}{2} \! - \! 4\epsilon \nonumber\\
4\epsilon  \! - \! \frac{1}{2} & \leq P(10|01)+P(01|10)+P(11|11)  \! - \! P(11|00)  \leq \frac{3}{2} \! - \! 4\epsilon  \nonumber\\
4\epsilon \! - \! \frac{1}{2} & \leq P(11|01)+P(11|10)+P(00|11) \! - \! P(11|00)  \leq \frac{3}{2} \! - \! 4\epsilon  \label{eq:ch}
\end{align}}

\noindent The relations~\eqref{eq:ch} are often called the CH inequalities~\cite{CH}, and it is convenient to work with these instead (note that the CH and CHSH inequalities can be seen as rewritings of the same inequality~\cite{Mermin95,Cereceda}, in particular, for each CHSH quantity there is a CH quantity, where the values are related by $J_{\chsh}^i(P)=\frac{3}{4}-\frac{1}{2}J_{\ch}^i(P)$|see e.g.~\cite{CC} for an explanation of the different ways to express the same Bell inequality).

Because the local state space is $\sta_{\Box}$, there are two extremal fiducial measurements on one gbit, ${M}_0$ and ${M}_1$, giving rise to the four combinations of these as the extremal fiducial measurements on two gbits. Hence, it is sufficient to restrict the CH values for these measurements, and to express every state in terms of the probabilities from these extremal fiducial measurements as a vector $V\left(S_{AB},\left({M}_0,{M}_1\right), \left({M}_0,{M}_1\right)\right)$.  Due to non-signalling and normalisation, each such vector is characterised by $8$ parameters. In this representation, the state space of a theory with restricted CHSH value is a convex subset of the polyhedron restricted by the non-signalling constraints, the positivity conditions, normalisation and the requirement that the relations~\eqref{eq:ch} hold. This polyhedron is the maximal possible state space and can be equivalently described in terms of its $80$ extremal vertices, which are found by identifying the vectors that simultaneously saturate $8$ of the facet-defining inequalities while obeying the remaining inequalities. We use $\mathcal{C}_\epsilon$ to denote the state space thus formed for a given $\epsilon\in[0,1/8]$. The conversion to the vertex picture was performed symbolically, i.e., with vertices expressed in terms of $\epsilon$, in {\sc Mathematica}. \bigskip

To approximate the effect space, we start by noting that a convex state space achieving the CH-value of $-\frac{1}{2}+4\epsilon$ (and $\frac{3}{2}-4\epsilon$ respectively) necessarily contains the state corresponding to the isotropic vector
\begin{align*}
  \left(
  \begin{array}{cc|cc}
    \epsilon & \frac{1}{2}-\epsilon & \epsilon & \frac{1}{2}-\epsilon \\
    \frac{1}{2}-\epsilon & \epsilon & \vphantom{\frac{1}{f}}\frac{1}{2}-\epsilon & \epsilon \\
    \hline
    \epsilon & \vphantom{\sum^5}\frac{1}{2}-\epsilon & \frac{1}{2}-\epsilon & \epsilon \\
    \frac{1}{2}-\epsilon & \epsilon & \epsilon & \frac{1}{2}-\epsilon \\
  \end{array}
  \right)
\end{align*}
(and $7$ others equivalent up to relabelling), where the entries represent the outcome probabilities for the fiducial measurements.  The reason for this is that using shared randomness and the symmetry of the state space we can form such a state by a depolarisation procedure (see e.g.,~\cite[Appendix~A]{Masanes2006}). 
The effect space of any theory with such a CH-value restriction is thus contained in the dual to the convex hull of the $16$ local deterministic vertices (the tensor products of the extremal states of two systems with local state space $\sta_\Box$) and the $8$ isotropic vectors with $\epsilon$ matching that of the state space under consideration. We use $\mathcal{D}_\epsilon$ to denote the corresponding convex cone of non-normalised effects, with origin at ${\bf{0}} \in \mathbb{R}^{16}$. It has $80$ extremal rays (plus $14$ that arise due to the redundancy in the description of the states, which, if normalisation is dropped, are characterised in terms of $9=16-7$  parameters). This cone was found by imposing that an effect, when applied to any of the above $24$ states, must yield a non-negative probability. The extremal rays of the effect cone $\mathcal{D}_\epsilon$ are then those that satisfy all of these $24$ inequalities and that give zero-probability on $8$ of the $24$, leaving only one parameter free. This was computed symbolically for any $\epsilon$ in {\sc Mathematica}.\footnote{When considering specific values of $\epsilon$, computational tools for vertex enumeration are available that allow for the computation of the dual cone to the non-normalised states in a much more efficient way.}

Hence, by construction, $\mathcal{C}_\epsilon$ and $\mathcal{D}_\epsilon$ are outer approximations to the true state and effect spaces of any theory with local state space $\sta_{\Box}$ in which two gbits have such a CHSH restriction.

Given a joint state space and a set of possible measurements, the maximum winning probability in any CHSH game is achieved with extremal local measurements\footnote{For an extremal effect $e$, the complementary effect $u-e$ is also extremal, making up an extremal measurement $\left\{e, u-e \right\}$ (see for instance~\cite{Kobayashi2017} for a proof).} acting on an extremal state. To assess the winning probability in the adaptive CHSH game we can think about the set of joint states that could form $S_{AC}$ under any effect applied on $BB'$. Due to the linearity of the winning condition, an upper bound on the winning probability of the adaptive CHSH game can be found by considering the extremal states in this set and the maximum CHSH value these can yield. (This is an upper bound because we only consider a single optimal effect applied on $BB'$, rather than a full measurement and it could be that there is no measurement whose effects all give rise to this winning probability.)  The set of extremal $S_{AC}$ can be formed by taking $S_{AB}$ and $S_{B'C}$ to be extremal states and considering extremal joint effects on $BB'$ (see Lemma~\ref{lem:postmeas} in Appendix~\ref{app:2}).

Using the above method and taking $\mathcal{C}_\epsilon$ for the extremal states and the cone $\mathcal{D}_\epsilon$ for the extremal effects, we can bound the winning probability in the adaptive CHSH game. These strategies can be enumerated by checking the winning probability for all states $S_{AC}$ obtained by sharing any two extremal states of $\mathcal{C}_\epsilon$ at the two sources and applying an effect on each of the extremal rays of $\mathcal{D}_\epsilon$ to $BB'$. Since we are considering a cone of non-normalised effects, the convex hull of all post-measurement states $S_{AC}$ that are obtained in this way, together with ${\bf{0}}$ form another convex cone. In order to derive bounds on the winning probability, only the (sub-)normalised states have to be considered, so we can truncate this cone at the vectors whose probabilities for each measurement sum to $1$. This is done by normalising all extremal post-measurement states, which span the truncated cone of all sub-normalised states. We note that, despite the incompatibility of $\mathcal{C}_\epsilon$ and $\mathcal{D}_\epsilon$, all of these strategies lead to valid probabilities. [One way to see this is to consider first applying the measurements on $A$ and $C$. This leads to a product state on the system $BB'$ (that depend on the outcomes). By construction, the effect on $BB'$ yields a valid probability when applied to a product state.]

In the above considerations, any separable extremal states and effects can be ignored and symmetries can be removed to reduce their number. Having done this, the number of remaining combinations is sufficiently small that they can all be computed symbolically in {\sc Mathematica}. This leads to the following bounds for the CH-value of the states $S_{AC}$ (for any $i\in\{1,2,3,4\}$):
\begin{equation}
\frac{2\epsilon(8\epsilon-1)}{96\epsilon^2-12\epsilon+1}\leq J^i_{\ch}(P)\leq \frac{80\epsilon^2-10\epsilon+1}{96\epsilon^2-12\epsilon+1}\,,
\end{equation}  
where $P$ is the distribution generated by $S_{AC}$ using the extremal fiducial measurements.

These bounds can be converted into bounds on the winning probability in the adaptive CHSH game as follows:
\begin{align*}
p_{\mathrm{win}}&\leq J^i_{\chsh}(P) \\
&=\frac{3}{4}-\frac{1}{2}J^i_{\ch}(P) \\
&\leq\frac{3}{4}+\frac{\epsilon(1-8\epsilon)}{96\epsilon^2-12\epsilon+1} \,.
\end{align*}
Similarly, we can lower bound $p_{\mathrm{win}}$ by
\begin{align}
p_{\mathrm{win}}&\geq\frac{3}{4}-\frac{80\epsilon ^2-10\epsilon+1}{96\epsilon^2-12\epsilon+1}.
\end{align}
\end{proof}

The bound~\eqref{eq:gbitbound} is maximal for $\epsilon=\frac{1}{16}$, which yields $p_{\mathrm{win}} \leq \frac{4}{5}$. Hence, no theory with single-system state space $\sta_{\Box}$ can recover the quantum performance in the adaptive CHSH game if the shared resources each comprise a single pair of gbits. Note that~\eqref{eq:gbitbound} may not be tight. In fact, when attempting to explicitly construct joint state spaces we were unable to beat the classical winning probability $\frac{3}{4}$.  In more detail, we considered firstly taking the joint space to be $\mathcal{C}_\epsilon$, i.e., where the extremal states are the $16$ local deterministic states, and noisy versions of the $8$ permutations of the PR-box with the joint effect space being the dual to this, and secondly a restriction of the state space of box-world in which we only allow all bipartite states that can achieve at most some fixed CHSH value, i.e., the state space is the dual to $\mathcal{D}_\epsilon$ and the effect space is $\mathcal{D}_\epsilon$. It is conceivable that the bound of $\frac{3}{4}$ also holds for any joint state space in between, whenever a valid corresponding effect space is considered.

The restriction of the sources to share one pair of gbits is imposed here to exclude the possibility of distilling stronger correlations after Bob performs his joint measurement. With the (slightly incompatible) approximations for state and effect spaces used in the proof, distillation of correlations would be possible (note that our outer approximation to the state space includes states that produce correlations that are distillable according to~\cite{Brunner2009}). Given this restriction on the dimension, we here solve the problem of dimension-dependent self-testing for $n=4$ in the adaptive CHSH game.

 However, it is not clear whether this restriction is necessary for the analysis. Partly, this is because, as mentioned above, the theories we could explicitly construct only led to a classical CHSH value. 
Furthermore, Bob can only output two bits $B$, which means that, even if distillation of correlations is in principle possible in a theory, we would need to know whether Alice and Charlie can distil sufficiently in this scenario in which they cannot communicate or share randomness, except via the shared resources that are established through Bob's operation.  In addition, we have derived our bounds by using only one effect for Bob. It may be the case that this effect cannot be completed to a four-outcome measurement, for which each effect allows us to achieve this bound, which is another reason why the true bound could be lower. Nonetheless, conclusively generalising our results to higher dimensional systems where distillation is possible requires further theory to be developed, as we explain in Section~\ref{sec:distill}.\bigskip

For a single system state space with an even number of extremal vertices greater than $4$ we encounter similar obstacles. It is known from~\cite{Janotta2011} that PR-box correlations can be distilled from the analogue of the maximally entangled state, $S_{\max}$, in any such theory with the maximal tensor product.  In these cases the set of correlations is strictly more non-local when we have access to several systems. This may not be the case for other tensor products though.  To illustrate this, consider the case where the local state space is a regular hexagon (i.e., $n=6$). We take a convex combination of $S_{\max}$ and the completely mixed state of the theory, $S_{\mathrm{mix}}$, to form $\epsilon S_{\max}+ (1-\epsilon)S_{\mathrm{mix}}$, and analyse the correlations that result when using the same local measurements as performed in~\cite{Janotta2011}.  Then~\cite{Forster2009} implies that for $\epsilon \geq 2/3$ the non-locality of these correlations, as measured by the maximal winning probability in the CHSH game, is bounded by those of the isotropic correlations with parameter $\epsilon'=1-4\frac{1-\epsilon}{2-\epsilon}$. For these correlations, the winning probability in the CHSH game is bounded by $\frac{1}{2}+\frac{\epsilon'+1}{4}$ (since they cannot be distilled with wirings~\cite{Beigi2015}). This indicates that the distillation protocols known for the maximal tensor product do not perform as well in other cases. Whether maximally non-local correlations can in these cases be distilled in other ways (e.g., from other correlations or in scenarios where the parties measure all subsystems they hold jointly) remains an open question.

\subsubsection{Self-dualized regular polygon state spaces with an even number of extremal vertices}\label{sec:sd}

The situation is considerably different if we \emph{self-dualize} these polygon systems~\cite{NoR2}. A self-dual theory has a (non-normalised) state and effect space that are formally represented by the same geometric object. Starting from any GPT, such as our regular polygon systems with an even number of vertices, we obtain a self-dual theory using  the following procedure that was introduced in~\cite{NoR2}. We first apply a scale transformation to the state space (along with the dual transformation for the effect space) such that the
state space is contained in the effect space after the transformation. This acts only as a change of coordinates, hence has no physical consequences. The effect space is subsequently reduced to coincide with the state space, in violation of the no-restriction hypothesis.

In this case, since the no-restriction hypothesis is not obeyed, we need to take a more general definition of the maximal tensor product into account~\cite{NoR2}, namely for local state spaces $\sta_A$, $\sta_B$ and local effect spaces $\eff_A$, $\eff_B$, the state space consists of all normalised bipartite states in 
$ (\eff_A \otimes_{\rm min} \eff_{\rm max}(\sta_B))^* \cap ( \eff_{\rm max}(\sta_A) \otimes_{\rm min} \eff_B)^*$, where $^*$ denotes the dual cone. This ensures that all effects act consistently on these states and that any post-measurement states are valid.  By using this joint state space, we find numerically, with a linear program analogous to the one used for odd polygon systems above, that Tsirelson's bound cannot be exceeded for $n=4,6,\ldots, 30$ with any state of two gbits in the maximal tensor product. We further find formulaic descriptions of the optimal behaviour that we conjecture to extend to higher even $n$, which we summarise in the following table. 
{\small
	\begin{center} \begin{tabular}{| c | c | c |} 
			\hline
			$n$ & CHSH-value & formulaic description\\
			\hline
			4 & 0.75 & $\frac{\vphantom{\sum^A}1}{\vphantom{\frac{1}{f}}2} +  \frac{\sqrt{2}}{4} \cos( \frac{\pi}{n})$  \\
			\hline
			6 & 0.8125 & $\frac{\vphantom{\sum^A}1}{\vphantom{\frac{1}{f}}2}  \! +  \! \frac{1}{8}[ 2 \cos(\frac{2+n}{4n}\pi)  \! -  \! \cos(\frac{3n-2}{4n}\pi)  \! +  \! \sin(\frac{6+n}{4n}\pi)]$ \\
			\hline
			8 & 0.8536 & $\frac{\vphantom{\sum^A}1}{\vphantom{\frac{1}{f}}2} + \frac{\sqrt{2}}{4}$ \\
			\hline
			10 & 0.8420 &  $\frac{\vphantom{\sum^A}1}{\vphantom{\frac{1}{f}}2} + \frac{3}{8} \sin(\frac{2+n}{4n}\pi) -\frac{1}{8}\cos(\frac{3n-6}{4n}\pi)$ \\
			\hline
	\end{tabular} \end{center} 
}

The upper bounds for the performance in the adaptive CHSH game were obtained by optimising the CHSH value over all bipartite states in the maximal tensor product of two such polygon systems. Performing this, we observe a similar periodicity with period $8$ as in the odd case above, which allowed us to derive similar formulaic descriptions of the optimal CHSH value. The table displays these by giving one particular $n$-value for each formula as well as the corresponding numeric CHSH value. We have certified the formulae to be accurate by running our optimisations up to $n=30$ and we conjecture their validity for any even $n$.

Taking their derivative, it is easy to see that each of these formulae is increasing in $n$ and converges to Tsirelson's bound in the limit $n \rightarrow \infty$. We also remark that for $n=8,16,24, \ldots$ we exactly recover Tsirelson's bound. This does not necessarily mean that there are theories with such local state spaces that recover the quantum performance in the adaptive CHSH game though: we know that with the maximal tensor product (which the bounds are derived from) this is definitely not possible and we do not know whether such correlations can even be achieved with smaller joint state spaces, let alone after performing a joint measurement like in the adaptive CHSH game. Thus we do not expect our upper bounds to be tight, i.e., the actual optimal performance in the adaptive CHSH game is likely worse.

Overall, similarly to the case of odd polygon systems, we have evidence that (subject to the restriction to one pair of gbits) self-dualized theories whose local state space is an even polygon (with any joint state space between minimal and maximal tensor product) cannot outperform quantum mechanics in the adaptive CHSH game. Since the odd polygon systems above were already self-dual, we can now make this statement for \emph{any} such self-dualized theory, independent of whether the number of extremal vertices is even or odd. We illustrate this in Figure~\ref{fig:plot}. Solving the dimension-independent case will require a better understanding of distillation (see Section~\ref{sec:distill}).

 \begin{figure}
 	\includegraphics[width=0.95\columnwidth]{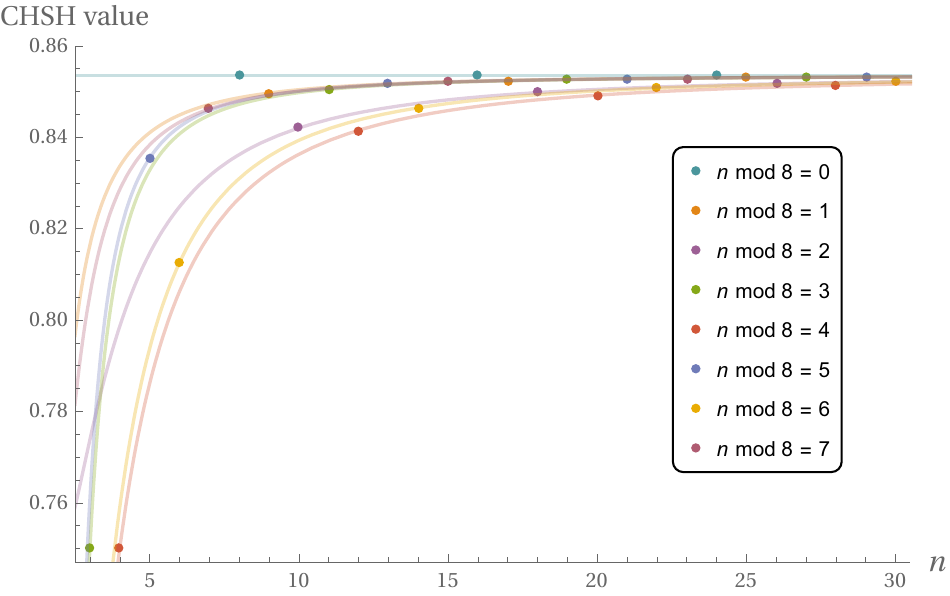}
 	\caption{Maximal CHSH value achievable in the maximal tensor product of self-dualized polygon systems. Along the horizontal axis we display the number of extremal vertices $n$, while on the vertical axis we display the respective optimal CHSH value. The points are the values obtained in our optimizations while the curves depict the respective formulae. The colours are used to group points that follow the same curve.}
 	\label{fig:plot}
 \end{figure}

Similar properties may hold for other local state spaces. For instance, in Section~6 of~\cite{Janotta2011}, it was shown that non-quantum correlations can be generated by local measurements on the analogue of the maximally entangled state when the local state space is `house-shaped', which is self-dual.\footnote{This state space is fairly unnatural, since it is not closed under permutations (see Section~\ref{sec:bipartite}).}  When the joint state space is the maximal tensor product, the authors of~\cite{Janotta2011} were unable to find any correlations violating Tsirelson's bound. Whether similar results hold for other theories with self-dual local state spaces and whether distillation of stronger correlations may be possible in such theories remains open.

\subsubsection{Relation to non-locality distillation and corresponding limitations}\label{sec:distill}
Most of the results in Section~\ref{sec:polygons} are valid under a dimension restriction on the shared systems or up to the distillation of non-locality. In order to solve the problem without the dimension restriction, we would need to better understand the set of correlations that can be obtained by Alice and Charlie, from the post-measurement states Bob can generate for them.

However, distillation of non-locality is not well-understood. While for distillation with local wirings (meaning performing local measurements on each subsystem of a bipartite system and then applying classical operations to the inputs and outputs) limits can be given~\cite{Forster2009, Beigi}, we are not aware of results on distillation with other (joint) measurements on different GPT systems, for instance. To solve our problem in general would require the theory of distillation to be developed.

\subsubsection{Relation to work on correlated boxes}

In recent years, a lot of work has been dedicated to the study of bipartite non-signalling correlations, without reference to any underlying theory. Such correlations can be modelled as behaviours of  black boxes where the two parties can make inputs and obtain (correlated) outputs.  In general, various underlying theories can lead to the same sets of such bipartite correlations (see for instance the even polygon systems with the maximal tensor product), even though the theories may look very different and may behave differently in other setups. The operations that can be performed locally by each party on their part of such a correlated box are restricted to choosing inputs (and processing the outputs obtained from the box).

In the case of the local state space $\sta_\Box$ we encounter the special situation that the correlated boxes that can be obtained by performing local measurements on a pair of gbits are in one-to-one correspondence with their joint states, since the extremal local effects are the unit vectors. Thus Theorem~\ref{thm:gbits} is directly related to research on correlated non-signalling boxes. In fact, for the local state space $\sta_\Box$, the set of all bipartite nonsignaling boxes with two inputs and two outputs is the maximal tensor product of two gbits.\footnote{For other local state spaces the maximal tensor product is different and may even lead to joint state spaces with restricted non-locality~\cite{Janotta2011}. See also the odd polygon systems considered here.}

In~\cite{Short2006,Skrzypczyk2009b,Skrzypczyk2009}, entanglement swapping of such boxes was analysed in terms of \emph{couplers}. Given two copies of a bipartite state, $S_{AB}$ and $S_{B'C}$, these are maps that act on $BB'$ and result in a copy of the same state shared on $AC$ (in some cases the copy is permitted to be noisy~\cite{Skrzypczyk2009}). Such couplers were shown not to exist if the bipartite state space comprises all non-signalling boxes~\cite{Short2006}.  However, they can be consistently defined if the bipartite state space is restricted to an asymmetric region of the non-signalling polytope (that is not closed under relabelling)~\cite{Skrzypczyk2009b, Skrzypczyk2009}.

Theorem~\ref{thm:gbits} contributes to this line of research in that it gives bounds on how non-local the state on $AC$ can be after non-locality swapping by applying linear maps (joint effects) where $S_{AB}$ and $S_{B'C}$ need not be the same state. Furthermore, our analysis is performed for theories with any amount of non-locality.  We note, however, that our consistency conditions (in particular, the symmetry of the state space under local permutations introduced at the end of Section~\ref{sec:bipartite}) differ from those of~\cite{Skrzypczyk2009, Skrzypczyk2009b}.

We also remark that analysing GPTs rather than correlated boxes has another advantage. Namely, when considering couplers that act consistently on some non-signalling boxes, but not on others, it is not so clear how to generalise this into a consistent theory of states and measurements, while a local state space and a composition operation, as considered in the GPT framework, immediately lead to sets of achievable correlated boxes. This is also especially important when aiming to understand various setups with different numbers of parties, where having an underlying theory is important for ensuring a consistent analysis.

\section{Conclusions} \label{sec:conclusion}
We have shown that in the adaptive CHSH game quantum mechanics outperforms the two-dimensional systems of all the theories that we were able to explicitly consider. Nonetheless, due to the need for further research on the distillation of correlations, we are not yet at the stage where this game can be fully analysed. In particular, it is desirable to generalise our results to allow the sharing of more general systems rather than gbit pairs and to take the possibility of distillation into account in the analysis.

Our optimal quantum strategy from Section~\ref{sec:quantum} is achievable by letting all states and measurements have real amplitudes and could thus also be achieved with quantum bits over a real Hilbert space. This is particularly interesting, since the corresponding local state space can be thought of as the $n \rightarrow \infty$ limit of a polygonal state space~\cite{Janotta2011}. Further research in this direction may lead us to identify a reason why the local state space of the smallest quantum system has to be spherical.

In addition, it is known that, while for any theory with a local state space coinciding with that of quantum theory the correlations for bipartite systems coincide with those of quantum mechanics~\cite{Barnum2010}, the same is not true for more than two parties~\cite{Acin2010}. Hence, some theories that only differ from quantum theory in their multi-partite correlations can only be self-tested in a setting that involves such multi-partite resources. Thus, in order to achieve complete correlation self-testing of quantum theory, adding a multi-partite task to the adaptive CHSH game might become necessary. In order for this to happen, one would first need to show that it is possible to formulate fully-fledged theories that lead to such sets of correlations and that recover the performance of quantum theory in the adaptive CHSH game.

It may turn out that there are numerous tasks for which quantum correlations are optimal. Identifying these and their general features is left as an interesting open question, the solution of which may lead to fundamental insights about quantum theory itself.

\bigskip

\acknowledgements This work was supported by an Engineering and Physical Sciences Research Council (EPSRC) First grant (grant number EP/P016588/1) and by the Austrian Science fund (FWF) stand-alone project P~30947.

% ------------------------------------------------------------------------------

%\bibliography{hplusgame}

%

% ------------------------------------------------------------------------------

\onecolumngrid
\appendix

\section{Proof of Proposition~\ref{thm:minmax}} \label{app:minmax}

For the proof of Proposition~\ref{thm:minmax}, we rely on the following lemma. 

\begin{lem}\label{lem:essential_entanglement}
In any GPT, the maximum probability of winning a CHSH game with a separable state is upper bounded by $3/4$.
\end{lem}

\begin{proof}
  A product measurement $M^{r_A,r_C}=\left\{e^{r_A}_a \otimes e^{r_C}_c \right\}_{a,c}$ on a separable state $S^{AC}=\sum_ip_iS^A_i\ot S^C_i,$ leads to probabilities
\begin{align}
p(a c | r_A r_C )_{P^{AC}} 
&=\sum_ip_ie^{r_A}_a(S^A_i) e^{r_C}_c(S^C_i)\\
&=\sum_ip_ip(a|r_A,i) p(c|r_C,i), \label{eq:distrsep}
\end{align}
where $p(a|r_A,i)=e^{r_A}_a(S^A_i)$ and $p(c|r_C,i)=e^{r_C}_c(S^C_i)$.  This distribution can be generated by sharing a source of classical randomness that outputs perfectly correlated values $i$ with probability $p_i$ to each party. Based on $i$ and their own inputs, $r_A$ and $r_B$ respectively, players $A$ and $C$ can sample from distributions $p(a|r_A,i)$ and $p(c|r_C,i)$.  The claim then follows because the winning probability of any classical strategy in any CHSH game is bounded by $3/4$~\cite{CHSH}.
\end{proof}

\begin{cor} \label{prop:minimal}
For any GPT where the joint state space is defined by the minimal tensor product, the winning probability in the adaptive CHSH game is upper bounded by $3/4$.
\end{cor}

%This follows immediately from the lack of entangled states in the minimal tensor product and Lemma~\ref{lem:essential_entanglement}.

\begin{lem}\label{lem:entangled_effects}
In any GPT, if Bob uses a separable effect, the states shared between Alice and Charlie in the adaptive CHSH game are separable.
\end{lem}

\begin{proof}
Applying a separable effect, i.e., $e^{BB'}_{\tilde{b}}=\sum_i s_i (e^{B}_{b_{i}} \otimes e^{B'}_{b'_{i}})$ with $s_i\geq0$ 
%and $\sum_i s_i=1$, 
to the product of two entangled states, we obtain 
\begin{align*}
S^{AC|\tilde{b}}&=\frac{1}{N} (I_A \otimes e_{\tilde{b}}^{BB'} \otimes I_C)\left(\sum_{j} p_{j} S^{AB}_j \ot S^{B'C}_j\right)\\
&=\frac{1}{N} \sum_{i,j} s_i p_{j} (I_A \otimes e^B_{b_i})(S^{AB}_j)\ot(e^{B'}_{b'_i} \otimes I_C)(S^{B'C}_j)\\
&=\sum_{i,j} \tilde{p}_{i,j} S^{\tilde{A}}_{i,j} \ot S^{\tilde{C}}_{i,j},
\end{align*}
with $N=e_{\tilde{b}}^{BB'}(S^{BB'})$ where $S^{BB'}$ is the reduced state of $S^{ABB'C}=\sum_jp_jS^{AB}_j\ot S^{B'C}_j$ on the subsystem $BB'$ (i.e., $S^{BB'}=(u_A\otimes I_B\otimes I_{B'}\ot u_C)S^{ABB'C}$); furthermore $S^{\tilde{A}}_{i,j}=(I_A \otimes e^B_{b_i})(S^{AB}_j) / e^{B}_{b_i}(S^{B}_j)$, $S^{\tilde{C}}_{i,j}=(e^{B'}_{b'_i} \ot I_C)(S^{B'C}_j) / e^{B'}_{b'_i}(S^{B'}_j)$ and $\tilde{p}_{i,j}=s_i p_j e^{B'}_{b'_i}(S^{B'}_j) e^{B}_{b_i}(S^{B}_j )/N$. Notice that $\tilde{p}_{i,j} \geq 0 \ \forall i,j$ and thus $\sum_{i,j}\tilde{p}_{i,j} =1$, as $S^{AC|\tilde{b}}$ is a state.
\end{proof}

\begin{cor} \label{prop:maximal}
For any GPT where the joint state space is defined by the maximal tensor product, the winning probability in the adaptive CHSH game is upper bounded by $3/4$.
\end{cor}

\noindent This follows immediately from the lack of entangled effects in the maximal tensor product by Lemma~\ref{lem:entangled_effects}.

Thus, according to Corollaries~\ref{prop:minimal} and~\ref{prop:maximal}, in a GPT where the joint state space is defined by the minimal or by the maximal tensor product no strategy can beat the winning probability of $3/4$.

\section{Post-measurement states on $AC$}\label{app:2}
\begin{lem} \label{lem:postmeas}
	In the adaptive CHSH game, the set of post-measurement states on the joint system $AC$ is within the convex hull of those states obtained by considering extremal joint states on $AB$ and $B'C$ and an extremal joint effect $e^{BB'}$. 
\end{lem}

\begin{proof}
  Any joint state on $AC$ obtained in the adaptive CHSH game can be written as
  \begin{equation}
    S^{AC}=\frac{(I_A\ot e^{BB'}\ot I_C) (S^{AB}\ot S^{B'C})}{e^{BB'}(S^B\ot S^{B'})},
  \end{equation}
  where $S^{B}$ and $S^{B'}$ are the reduced states of $S^{AB}$ and $S^{B'C}$ respectively.  All states $S^{AB}\ot S^{B'C}$ live in the state space
  \begin{equation}
    \sta^{ABB'C}=\operatorname{conv} \left(\left\{S^{AB} \ot S^{B'C} \mid S^{AB} \in \ext(\sta^{AB}), S^{B'C} \in \ext(\sta^{B'C})\right\} \right),
  \end{equation}
  where $\operatorname{conv}$ denotes the convex hull and $\ext(\sta)$ represents the set of extremal states of $\sta$.  Since $I_A\ot e^{BB'}\ot I_C$ is linear, the extremal states that can be obtained upon application of $I_A\ot e^{BB'}\ot I_C$ to states from the convex set $\sta^{ABB'C}$ are all images of extremal states. This also holds if we normalise all post-measurement states, as can be checked with a straightforward calculation. Hence, each extremal $e^{BB'}$ allows us to construct a convex set by considering its application to all extremal states of $\sta^{ABB'C}$. The union of these sets is within the convex hull of all their extremal points.
	
  Furthermore, for any non-extremal effect $\hat{e}^{BB'}=p e^{BB'} +(1-p) \tilde{e}^{BB'}$, the normalised post-measurement state is
  \begin{align}
    \hat{S}^{AC}
         &= \frac{ p (I_A \ot e^{BB'} \ot I_C) (S^{AB} \ot S^{B'C})+(1-p) (I_A\ot \tilde{e}^{BB'}\ot I_C) (S^{AB} \ot S^{B'C})}{p e^{BB'} (S^B\ot S^{B'})+(1-p)  \tilde{e}^{BB'} (S^B\ot S^{B'})}  \\
         &=\tilde{p} S^{AC} + (1-\tilde{p}) \tilde{S}^{AC},
  \end{align} 
  where $\tilde{p}= (p e^{BB'} (S^B\ot S^{B'})) / ({p e^{BB'} (S^B\ot S^{B'})+(1-p) \tilde{e}^{BB'} (S^B\ot S^{B'})})$.
	
  Thus, the set of post-measurement states is within the convex hull of all post-measurement states that are obtained from extremal effects $e^{BB'}$ applied to the extremal states of $\sta^{ABB'C}$.
\end{proof}

\end{document}